%% file: paper.tex
\RequirePackage{amsmath}
\documentclass[12pt]{article}

\usepackage{booktabs}

\newcommand\samethanks[1][\value{footnote}]{\footnotemark[#1]}

\usepackage[top=1in,bottom=1in,left=1in,right=1in]{geometry}
\makeatletter
\def\myfnt{\ifx\protect\@typeset@protect\expandafter\footnote\else\expandafter\@gobble\fi}
\makeatother

\makeatletter
\def\BState{\State\hskip-\ALG@thistlm}
\makeatother



\makeatother

\input{macros}

\usetikzlibrary{calc}
\begin{document}

\title{Online Diverse Learning to Rank from Partial-Click Feedback}
\author{Prakhar Gupta\thanks{Equal Contribution} \,\thanks{Carnegie Mellon University}\\
prakharg@cs.cmu.edu
\and
Gaurush Hiranandani\samethanks[1] \,\thanks{University of Illinois Urbana-Champaign}\\
gaurush2@illinois.edu
\and
Harvineet Singh\samethanks[1] \,\thanks{New York University} \\
hs3673@nyu.edu
\and
Branislav Kveton\thanks{Google Research} \\
bkveton@google.com
\and
Zheng Wen\thanks{Adobe Research}\\
zwen@adobe.com
\and
Iftikhar Ahamath Burhanuddin\samethanks[6] \\
burhanud@adobe.com 
}
\date{\today}

\flushbottom
\maketitle

\input{Abstract}
%
\input{Introduction}
%
\input{Background}
%
\input{Model}
%
\input{Algorithm}
%
\input{Analysis}
%
\input{Experiments}
%
\input{RelatedWork}
%
\input{Conclusions}
\bibliographystyle{plain}
\bibliography{References}
\input{proof}
	
\end{document}

%% file: macros.tex
\usepackage{graphicx}
\usepackage{caption}
\usepackage{algorithm}
\usepackage[noend]{algpseudocode}
\usepackage{amssymb}
\usepackage{amsopn}
\usepackage{bbm}
\usepackage{bm}
\usepackage{color}
\usepackage{dsfont}
\usepackage{enumerate}
\usepackage{subfigure}
\usepackage{xcolor}

\usepackage{booktabs} 

\usepackage[backgroundcolor=White]{todonotes}
\newcommand{\commentout}[1]{}
\newcommand{\junk}[1]{}
\newcommand{\etalabbr}{\emph{et al.}}

\newcommand{\cE}{\mathcal{E}}

\newcommand{\cH}{\mathcal{H}}
\newcommand{\cG}{\mathcal{G}}

\newtheorem{remark}{Remark}
\newtheorem{theorem}{Theorem}

\newtheorem{lemma}{Lemma}
\newtheorem{proof}{Proof}

\newcommand{\abs}[1]{\left|#1\right|}

\newcommand{\E}[1]{\mathbb{E} \left[#1\right]}

\newcommand{\I}[1]{\mathds{1} \! \left\{#1\right\}}

\newcommand{\rnd}[1]{#1}
\newcommand{\set}[1]{\left\{#1\right\}}

\newcommand{\transpose}{^\mathsf{\scriptscriptstyle T}}

\DeclareMathOperator*{\argmax}{arg\,max\,}

\mathchardef\mhyphen="2D

\newcommand{\cascadeklucb}{{\tt CascadeKL\mhyphen UCB}}
\newcommand{\cascadelinucb}{{\tt CascadeLinUCB}}
\newcommand{\cascadelsb}{{\tt CascadeLSB}}

\newcommand{\lsbgreedy}{{\tt LSBGreedy}}

\usepackage[capitalize,noabbrev]{cleveref}

%% file: Abstract.tex
\begin{abstract}
Learning to rank is an important problem in machine learning and recommender systems. In a recommender system, a user is typically recommended a list of items. Since the user is unlikely to examine the entire recommended list, partial feedback arises naturally. At the same time, diverse recommendations are important because it is challenging to model all tastes of the user in practice. In this paper, we propose the first algorithm for online learning to rank diverse items from partial-click feedback. We assume that the user examines the list of recommended items until the user is attracted by an item, which is clicked, and does not examine the rest of the items. This model of user behavior is known as the cascade model. We propose an online learning algorithm, $\cascadelsb$, for solving our problem. The algorithm actively explores the tastes of the user with the objective of learning to recommend the optimal diverse list. We analyze the algorithm and prove a gap-free upper bound on its $n$-step regret. We evaluate $\cascadelsb$ on both synthetic and real-world datasets, compare it to various baselines, and show that it learns even when our modeling assumptions do not hold exactly. 
\end{abstract}

%% file: Introduction.tex

\section{Introduction}
\label{sec:introduction}

Learning to rank is an important problem with many practical applications in web search \cite{agichtein06improving}, information retrieval \cite{liu2009learning}, and recommender systems \cite{ricci11introduction}. Recommender systems typically recommend a list of items, which allows the recommender to better cater to the various tastes of the user by the means of diversity \cite{carbonell1998use,mei2010divrank,adomavicius2012improving}. In practice, the user rarely examines the whole recommended list, because the list is too long or the user is satisfied by some higher-ranked item. These aspects make the problem of learning to rank from user feedback challenging.

User feedback, such as clicks, have been found to be an informative source for training learning to rank models \cite{agichtein06improving,yu2015learning}. This motivates recent research on online learning to rank \cite{radlinski08learning,slivkins13ranked}, where the goal is to interact with the user to collect clicks, with the objective of learning good items to recommend over time. Online learning to rank methods have shown promising results when compared to state-of-the-art offline methods \cite{grotov2016online}. This is not completely unexpected. The reason is that offline methods are inherently limited by past data, which are biased due to being collected by production policies. Online methods overcome this bias by sequential experimentation with users.

The concept of diversity was introduced in online learning to rank in ranked bandits \cite{radlinski08learning}. Yue and Guestrin \cite{guestrin2011nips} formulated this problem as learning to maximize a submodular function and proposed a model that treats clicks as cardinal utilities of items. Another line of work is by Raman et al.~\cite{raman2012online}, who instead of treating clicks as item-specific cardinal utilities, only rely on preference feedback between two rankings. However, these past works on diversity do not address a common bias in learning to rank, which is that lower-ranked items are less likely to be clicked due to the so-called \emph{position bias}. One way of explaining this bias is by the cascade model of user behavior \cite{craswell08experimental}. In this model, the user examines the recommended list from the first item to the last, clicks on the first attractive item, and then leaves without examining the remaining items. This results in  \emph{partial-click feedback}, because we do not know whether the items after the first attractive item would be clicked. Therefore, clicks on lower-ranked items are biased due to higher-ranked items. Although the cascade model may seem limited, it has been found to be extremely effective in explaining user behavior \cite{chuklin2015click}. Several recent papers proposed online learning to rank algorithms in the cascade model \cite{kveton15cascading,kveton15combinatorial,zong16cascading,li2016contextual}. None of these papers recommend a diverse list of items.

In this paper, we propose a novel approach to online learning to rank from clicks that addresses both the challenges of partial-click feedback due to position bias and diversity in the recommendations. In particular, we make the following contributions:
\begin{itemize}
  \item We propose a \emph{diverse cascade model}, a novel click model that models both partial-click feedback and diversity.
  \item We propose a \emph{diverse cascading bandit}, an online learning framework for learning to rank in the diverse cascade model.
  \item We propose $\cascadelsb$, a computationally-efficient algorithm for learning to rank in the diverse cascading bandit. We analyze this algorithm and derive a $O(\sqrt{n})$ upper bound on its $n$-step regret.
  \item We comprehensively evaluate $\cascadelsb$ on one synthetic problem and three real-world problems, and show that it consistently outperforms baselines, even when our modeling assumptions do not hold exactly.
\end{itemize}

The paper is organized as follows. In \cref{sec:background}, we present the necessary background to understand our work. We propose our new diverse cascade model in \cref{sec:model}. In \cref{sec:bandit}, we formulate the problem of online learning to rank in our click model. Our algorithm, $\cascadelsb$, is proposed in \cref{sec:algorithm} and analyzed in \cref{sec:analysis}. In \cref{sec:experiments}, we empirically evaluate $\cascadelsb$ on several problems. We review related work in \cref{sec:related work}. Conclusions and future work are discussed in \cref{sec:conclusions}.

We define $[n] = \set{1, \dots, n}$. For any sets $A$ and $B$, we denote by $A^B$ the set of all vectors whose entries are indexed by $B$ and take values from $A$. We treat all vectors as column vectors. 

%% file: Background.tex

\section{Background}
\label{sec:background}

This section reviews two \emph{click models} \cite{chuklin2015click}. A click model is a stochastic model that describes how the user interacts with a list of items. More formally, let $E = [L]$ be a \emph{ground set} of $L$ items, such as the set of all web pages or movies. Let $A = (a_1, \dots, a_K) \in \Pi_K(E)$ be a \emph{list of $K\leq L$ recommended items}, where $a_k$ is the $k$-th recommended item and $\Pi_K(E)$ is the set of all \emph{$K$-permutations} of the set $E$. Then the click model describes how the user \emph{examines} and \emph{clicks} on items in any list $A$.

\subsection{Cascade Model}
\label{sec:cascade model}

The \emph{cascade model} \cite{craswell08experimental} explains a common bias in recommending multiple items, which is that lower ranked items are less likely to be clicked than higher ranked items. The model is parameterized by $L$ \emph{item-dependent attraction probabilities} $\bar{w} \in [0, 1]^E$.

The user examines a recommended list $A \in \Pi_K(E)$ from the first item $a_1$ to the last $a_K$. When the user examines item $a_k$, the item attracts the user with probability $\bar{w}(a_k)$, independently of the other items. If the user is attracted by an item $a_k$, the user clicks on it and does not examine any of the remaining items. If the user is not attracted by item $a_k$, the user examines the next item $a_{k + 1}$. The first item is examined with probability one.

Since each item attracts the user independently, the probability that item $a_k$ is examined is $\prod_{i = 1}^{k - 1} (1 - \bar{w}(a_i))$, and the probability that at least one item in $A$ is attractive is
\begin{align}
  \textstyle
  1 - \prod_{k = 1}^K (1 - \bar{w}(a_k))\,.
  \label{eq:cascade click probability}
\end{align}
Clearly, this objective is maximized by $K$ most attractive items.

\subsection{Diverse Click Model - Yue and Guestrin \cite{guestrin2011nips}}
\label{sec:yue model}

Submodularity is well established in diverse recommendations \cite{carbonell1998use}. In the model of Yue and Guestrin \cite{guestrin2011nips}, the probability of clicking on an item depends on the gains in topic coverage by that item and the interests of the user in the covered topics.

Before we discuss this model, we introduce basic terminology. The \emph{topic coverage} by items $S \subseteq E$, $c(S) \in [0, 1]^d$, is a $d$-dimensional vector whose $j$-th entry is the coverage of topic $j \in [d]$ by items $S$. In particular, $c(S) = (c_1(S), \dots, c_d(S))$, where $c_j(S)$ is a \emph{monotone} and \emph{submodular} function in $S$ for all $j \in [d]$, that is
\begin{align*}
  \forall A \subseteq E, e \in E: & \ c_j(A \cup \set{e}) \geq c_j(A)\,, \\
  \forall A \subseteq B \subseteq E, e \in E: &
  \ c_j(A \cup \set{e}) - c_j(A) \geq c_j(B \cup \set{e}) - c_j(B)\,.
\end{align*}
For topic $j$, $c_j(S) = 0$ means that topic $j$ is not covered at all by items $S$, and $c_j(S) = 1$ means that topic $j$ is completely covered by items $S$. For any two sets $S$ and $S'$, if $c_j(S) > c_j(S')$, items $S$ cover topic $j$ better than items $S'$.

The \emph{gain in topic coverage} by item $e$ over items $S$ is defined as
\begin{align}
  \Delta(e \mid S) = c(S \cup \set{e}) - c(S)\,.
  \label{eq:topic coverage gain}
\end{align}
Since $c_j(S) \in [0, 1]$ and $c_j(S)$ is monotone in $S$, $\Delta(e \mid S) \in [0, 1]^d$. The \emph{preferences} of the user are a distribution over $d$ topics, which is represented by a vector $\theta^\ast = (\theta^\ast_1, \dots, \theta^\ast_d)$.

The user examines all items in the recommended list $A$ and is attracted by item $a_k$ with probability
\begin{align}
  \langle\Delta(a_k \mid \set{a_1, \dots, a_{k - 1}}), \theta^\ast\rangle\,,
  \label{eq:attraction probability}
\end{align}
where $\langle\cdot,\cdot\rangle$ is the dot product of two vectors. The quantity in \eqref{eq:attraction probability} is the gain in topic coverage after item $a_k$ is added to the first $k - 1$ items weighted by the preferences of the user $\theta^\ast$ over the topics. Roughly speaking, if item $a_k$ is diverse over higher-ranked items in a topic of user's interest, then that item is likely to be clicked.

If the user is attracted by item $a_k$, the user clicks on it. It follows that the expected number of clicks on list $A$ is $\langle c(A), \theta^\ast\rangle$, where
\begin{align}
  \langle c(A), \theta\rangle =
  \sum_{k = 1}^K \langle\Delta(a_k \mid \set{a_1, \dots, a_{k - 1}}), \theta\rangle
  \label{eq:yue number of clicks}
\end{align}
for any list $A$, preferences $\theta$, and topic coverage $c$.

%% file: Model.tex

\section{Diverse Cascade Model}
\label{sec:model}

Our work is motivated by the observation that none of the models in \cref{sec:background} explain both the position bias and diversity together. The optimal list in the cascade model are $K$ most attractive items (\cref{sec:cascade model}). These items are not guaranteed to be diverse, and hence the list may seem repetitive and unpleasing in practice. The diverse click model in \cref{sec:yue model} does not explain the position bias, that lower ranked items are less likely to be clicked. We illustrate this on the following example. Suppose that item $1$ completely covers topic $1$, $c(\set{1}) = (1, 0)$, and that all other items completely cover topic $2$, $c(\set{e}) = (0, 1)$ for all $e \in E \setminus \set{1}$. Let $c(S) = \max_{e \in S} c(\set{e})$, where the maximum is taken entry-wise, and $\theta^\ast = (0.5, 0.5)$. Then, under the model in \cref{sec:yue model}, item $1$ is clicked with probability $0.5$ in any list $A$ that contains it, irrespective of its position. This means that the position bias is not modeled well.

\subsection{Click Model}

We propose a new click model, which addresses both the aforementioned phenomena, diversity and position bias. The diversity is over $d$ topics, such as movie genres or restaurant types. The preferences of the user are a distribution over these topics, which is represented by a vector $\theta^\ast = (\theta^\ast_1, \dots, \theta^\ast_d)$.

We refer to our model as a \emph{diverse cascade model}. The user interacts in this model as follows. The user scans a list of $K$ items $A = (a_1, \dots, a_K) \in \Pi_K(E)$ from the first item $a_1$ to the last $a_K$, as described in \cref{sec:cascade model}. If the user examines item $a_k$, the user is attracted by it proportionally to its gains in topic coverage over the first $k - 1$ items weighted by the preferences of the user $\theta^\ast$ over the topics. The attraction probability of item $a_k$ is defined in \eqref{eq:attraction probability}. Roughly speaking, if item $a_k$ is diverse over higher-ranked items in a topic of user's interest, then that item is likely to attract the user. If the user is attracted by item $a_k$, the user clicks on it and does not examine any of the remaining items. If the user is not attracted by item $a_k$, then the user examines the next item $a_{k + 1}$. The first item is examined with probability one.

We assume that each item attracts the user independently, as in the cascade model (\cref{sec:cascade model}). Under this assumption, the probability that at least one item in $A$ is attractive is $f(A, \theta^\ast)$, where
\begin{align}
  f(A, \theta) =
  1 - \prod_{k = 1}^K (1 - \langle\Delta(a_k \mid \set{a_1, \dots, a_{k - 1}}), \theta\rangle)
  \label{eq:click probability}
\end{align}
for any list $A$, preferences $\theta$, and topic coverage $c$.

\subsection{Optimal List}
\label{sec:optimal list}

To the best of our knowledge, the list that maximizes \eqref{eq:click probability} under user preferences $\theta^\ast$,
\begin{align}
  \textstyle
  A^\ast = \argmax_{A \in \Pi_K(E)} f(A, \theta^\ast)\,,
  \label{eq:optimal list}
\end{align}
cannot be computed efficiently. Therefore, we propose a greedy algorithm that maximizes $f(A, \theta^\ast)$ approximately. The algorithm chooses $K$ items sequentially. The $k$-th item $a_k$ is chosen such that it maximizes its gain over previously chosen items $a_1, \dots, a_{k - 1}$. In particular, for any $k \in [K]$,
\begin{align}
  \textstyle
  a_k = \argmax_{e \in E \setminus \set{a_1, \dots, a_{k - 1}}}
  \langle\Delta(e \mid \set{a_1, \dots, a_{k - 1}}), \theta^\ast\rangle\,.
  \label{eq:greedy maximization}
\end{align}
We would like to comment on the quality of the above approximation. First, although the value of adding any item $e$ diminishes with more previously added items, $f(A, \theta^\ast)$ is not a set function of $A$ because its value depends on the order of items in $A$. Therefore, we do not maximize a monotone and submodular set function, and thus we do not have the well-known $1 - 1 / e$ approximation ratio \cite{nemhauser78approximation}. Nevertheless, we can still establish the following guarantee.

\begin{theorem}
\label{thm:approximation_ratio} For any topic coverage $c$ and user preferences $\theta^*$, let $A^{\text{greedy}}$ be the solution computed by the greedy algorithm in \eqref{eq:greedy maximization}. Then
\begin{align}
  \!\!\!
  f(A^{\text{greedy}}, \theta^*) \geq (1 - 1 / e)
  \max \left\{\frac{1}{K}, 1 - \frac{K - 1}{2} c_{\max}\right\} f(A^\ast, \theta^\ast)\,,
\end{align}
where $c_{\max} = \max_{e \in E} \langle c(\{e\}), \theta^*\rangle$ is the maximum click probability. In other words, the approximation ratio of the greedy algorithm is $ (1 - 1 / e) \max \left \{\frac{1}{K}, 1- \frac{K-1}{2} c_{\max}\right\}$.
\end{theorem}
\begin{proof}
The proof is in \cref{sec:approximation_ratio}.
\end{proof}

Note that when $c_{\max}$ in \cref{thm:approximation_ratio} is small, the approximation ratio is close to $1 - 1 / e$. This is common in ranking problems where the maximum click probability $c_{\max}$ tends to be small. In \cref{sec:approximation ratio}, we empirically show that our approximation ratio is close to one in practice, which is significantly better than suggested in \cref{thm:approximation_ratio}.



\section{Diverse Cascading Bandit}
\label{sec:bandit}

In this section, we present an online learning variant of the diverse cascade model (\cref{sec:model}), which we call a \emph{diverse cascading bandit}. An instance of this problem is a tuple $(E, c, \theta^\ast, K )$, where $E = [L]$ represents a ground set of $L$ items, $c$ is the topic coverage function in \cref{sec:yue model}, $\theta^\ast$ are user preferences in \cref{sec:model}, and $K \leq L$ is the number of recommended items. The preferences $\theta^\ast$ are unknown to the learning agent.

Our learning agent interacts with the user as follows. At time $t$, the agent recommends a list of $K$ items $\rnd{A}_t = (\rnd{a}^t_1, \dots, \rnd{a}^t_K) \in \Pi_K(E)$. The attractiveness of item $a_k$ at time $t$, $\rnd{w}_t(\rnd{a}^t_k)$, is a realization of an independent Bernoulli random variable with mean $\langle\Delta(\rnd{a}^t_k \mid \{\rnd{a}^t_1, \dots, \rnd{a}^t_{k - 1}\}), \theta^\ast\rangle$. The user examines the list from the first item $\rnd{a}^t_1$ to the last $\rnd{a}^t_K$ and clicks on the first attractive item. The \emph{feedback} is the index of the click, $\rnd{C}_t = \min \set{k \in [K]: \rnd{w}_t(\rnd{a}^t_k) = 1}$, where we assume that $\min \emptyset = \infty$. That is, if the user clicks on an item, then $\rnd{C}_t \leq K$; and if the user does not click on any item, then $\rnd{C}_t = \infty$. We say that item $e$ is \emph{examined} at time $t$ if $e = \rnd{a}^t_k$ for some $k \in [\min \set{\rnd{C}_t, K}]$. Note that the attractiveness of all examined items at time $t$ can be computed from $\rnd{C}_t$. In particular, $\rnd{w}_t(\rnd{a}^t_k) = \I{\rnd{C}_t = k}$ for any $k \in [\min \set{\rnd{C}_t, K}]$. The \emph{reward} is defined as $\rnd{r}_t = \I{\rnd{C}_t \leq K}$. That is, the reward is one if the user is attracted by at least one item in $\rnd{A}_t$; and zero otherwise.


The goal of the learning agent is to maximize its expected cumulative reward. This is equivalent to minimizing the expected cumulative regret with respect to the \emph{optimal list} in \eqref{eq:optimal list}. The regret is formally defined in \cref{sec:analysis}.

%% file: Algorithm.tex

\section{Algorithm $\cascadelsb$}
\label{sec:algorithm}

\begin{algorithm}[t]
  \caption{$\cascadelsb$ for solving diverse cascading bandits.}
  \label{alg:main}
  \begin{algorithmic}[1]
    \State \textbf{Inputs:} Tunable parameters $\sigma>0$ and $\alpha>0$ (\cref{sec:analysis})
    \Statex \vspace{-0.05in}
    \State $\rnd{M}_0 \gets I_d$, \ $\rnd{B}_0 \gets \mathbf{0}$
    \Comment{Initialization}
    \For{$t = 1, \dots, n$}
      \State $\bar{\theta}_{t - 1} \gets \sigma^{-2} \rnd{M}_{t - 1}^{-1} \rnd{B}_{t - 1}$
      \Comment{Regression estimate of $\theta^\ast$}
      \Statex \vspace{-0.05in}
      \State $S \gets \emptyset$
      \Comment{Recommend list $A_t$ and receive feedback $C_t$}
      \For{$k = 1, \dots, K$}
        \ForAll{$e \in E \setminus S$}
          \State $x_e \gets \Delta(e \mid S)$
        \EndFor
        \State $\displaystyle \rnd{a}^t_k \gets \argmax_{e \in E \setminus S}
        \left[x_e\transpose \bar{\theta}_{t - 1} + \alpha \sqrt{x_e\transpose \rnd{M}_{t - 1}^{-1} x_e}\right]$
        \State $S \gets S \cup \set{\rnd{a}^t_k}$
      \EndFor
      \State Recommend list $\rnd{A}_t \gets (\rnd{a}^t_1, \dots, \rnd{a}^t_K)$
      \State Observe click $\rnd{C}_t \in \set{1, \dots, K, \infty}$
      \Statex \vspace{-0.05in}
      \State $\rnd{M}_t \gets \rnd{M}_{t - 1}$, \ $\rnd{B}_t \gets \rnd{B}_{t - 1}$
      \Comment{Update statistics}
      \For{$k = 1, \dots, \min \set{\rnd{C}_t, K}$}
        \State $x_e \gets \Delta\left(\rnd{a}^t_k \ \middle| \ \set{\rnd{a}^t_1, \dots, \rnd{a}^t_{k - 1}}\right)$
        \State $\rnd{M}_t \gets \rnd{M}_t + \sigma^{-2} x_e x_e\transpose$
        \State $\rnd{B}_t \gets \rnd{B}_t + x_e \I{\rnd{C}_t = k}$
      \EndFor
	\EndFor
  \end{algorithmic}
\end{algorithm}

Our algorithm for solving diverse cascading bandits is presented in \cref{alg:main}. We call it $\cascadelsb$, which stands for a \emph{cascading linear submodular bandit}. We choose this name because the attraction probability of items is a linear function of user preferences and a submodular function of items.

The algorithm knows the gains in topic coverage $\Delta(e \mid S)$ in \eqref{eq:topic coverage gain}, for any item $e \in E$ and set $S \subseteq E$. It does not know the user preferences $\theta^\ast$ and estimates them through repeated interactions with the user. It also has two tunable parameters $\sigma>0$ and $\alpha>0$, where $\sigma$ controls the growth rate of the Gram matrix (line $16$) and $\alpha$ controls the degree of optimism (line $9$).

At each time $t$, $\cascadelsb$ has three stages. In the first stage (line $4$), we estimate $\theta^\ast$ as $\bar{\theta}_{t - 1}$ by solving a least-squares problem. Specifically, we take all observed topic gains and responses up to time $t$, which are summarized in $\rnd{M}_{t - 1}$ and $\rnd{B}_{t - 1}$, and then estimate $\bar{\theta}_{t - 1}$ that fits these responses the best.

In the second stage (lines $5$--$12$), we recommend the best list of items under $\bar{\theta}_{t - 1}$ and $\rnd{M}_{t - 1}$. This list is generated by the greedy algorithm from \cref{sec:optimal list}, where the attraction probability of item $e$ is overestimated as $x_e\transpose \bar{\theta}_{t - 1} + \alpha \sqrt{x_e\transpose \rnd{M}_{t - 1}^{-1} x_e}$ and $x_e$ is defined in line $8$. This optimistic overestimate is known as the \emph{upper confidence bound (UCB)} \cite{auer02finitetime}. 

In the last stage (lines $13$--$17$), we update the Gram matrix $\rnd{M}_t$ by the outer product of the observed topic gains, and the response matrix $\rnd{B}_t$ by the observed topic gains weighted by their clicks.

The time complexity of each iteration of $\cascadelsb$ is $\mathcal{O}(d^3 + K L d^2)$. The Gram matrix $\rnd{M}_{t - 1}$ is inverted in $\mathcal{O}(d^3)$ time. The term $K L d^2$ is due to the greedy maximization, where we select $K$ items out of $L$, based on their UCBs, each of which is computed in $\mathcal{O}(d^2)$ time. The update of statistics takes $\mathcal{O}(Kd^2)$ time. $\cascadelsb$ takes $\mathcal{O}(d^2)$ space due to storing the Gram matrix $M_t$.

%% file: Analysis.tex

\section{Analysis}
\label{sec:analysis}

Let $\gamma = (1 - 1 / e) \max \left\{\frac{1}{K}, 1 - \frac{K - 1}{2} c_{\max}\right\}$ and $A^\ast$ be the optimal solution in \eqref{eq:optimal list}. Then based on \cref{thm:approximation_ratio}, $A_t$ (line $11$ of \cref{alg:main}) is a $\gamma$-approximation at any time $t$. Similarly to Chen \etalabbr~\cite{Chen:2016:CMB:2946645.2946695}, Vaswani \etalabbr~\cite{pmlr-v70-vaswani17a}, and Wen \etalabbr~\cite{NIPS2017_6895}, we define the \emph{$\gamma$-scaled $n$-step regret} as
\begin{equation}
  R^{\gamma}(n) = \sum_{t=1}^n \E{  f(A^*, \theta^\ast) - f(A_t, \theta^\ast)/\gamma}{},
  \label{analysis:cumulative_regret}
\end{equation}
where the scaling factor $1 / \gamma$ accounts for the fact that $A_t$ is a $\gamma$-approximation at any time $t$. This is a natural performance metric in our setting. This is because even the offline variant of our problem in \eqref{eq:optimal list} cannot be solved optimally computationally efficiently. Therefore, it is unreasonable to assume that an online algorithm, like $\cascadelsb$, could compete with $A^\ast$. Under the scaled regret, $\cascadelsb$ competes with comparable computationally-efficient offline approximations. Our main theoretical result is below.

\begin{theorem}
\label{theorem:main}
Under the above assumptions, for any $\sigma>0$ and any 
\begin{align}
\label{analysis:c}
\alpha \geq \frac{1}{\sigma} \sqrt{d \log \left( 1 + \frac{nK}{d \sigma^2}\right) + 2 \log \left( n \right)} + \|\theta^* \|_2
\end{align}
in \cref{alg:main}, where $\|\theta^* \|_2 \leq \|\theta^* \|_1 \leq 1$, we have
\begin{equation}
\label{analysis:regret_bound}
R^{\gamma}(n) \leq \frac{2\alpha K}{\gamma}   \sqrt{\frac{dn \log \left[ 1 + \frac{nK}{d \sigma^2}\right]}{\log \left( 1+ \frac{1}{\sigma^2}\right)}}+1.
\end{equation}
\end{theorem}
\begin{proof}
The proof is in \cref{sec:proof}.
\end{proof}

\noindent
\cref{theorem:main} states that for any $\sigma>0$ and a sufficiently optimistic $\alpha$ for that $\sigma$, the regret bound in \eqref{analysis:regret_bound} holds. Specifically, if we choose $\sigma=1$ and
\[
\alpha = \sqrt{d \log \left( 1 + \frac{nK}{d }\right) + 2 \log \left( n \right)} + \eta
\]
for some $\eta \geq \|\theta^* \|_2$, then $R^{\gamma}(n) = \tilde{O} \left( dK \sqrt{n}/\gamma \right)$, where the
$\tilde{O}$ notation hides logarithmic factors. We now briefly discuss the tightness of this bound. The $\tilde{O} \left( \sqrt{n} \right)$-dependence on the time horizon $n$ is considered near-optimal in gap-free regret bounds. The $\tilde{O}(d)$-dependence on the number of features $d$ is standard in linear bandits~\cite{abbasi-yadkori11improved}. As we discussed above, the $O(1 / \gamma)$ factor is due to the fact that $A_t$ is a $\gamma$-approximation. The $\tilde{O}(K)$-dependence on the number of recommended items $K$ is due to the fact that the agent recommends $K$ items. We believe that this dependence can be reduced to $\tilde{O}(\sqrt{K})$ by a better analysis. We leave this for future work.

Finally, note that the list $\rnd{A}_t$ in $\cascadelsb$ is constructed greedily. However, our regret bound in \cref{theorem:main} does not make any assumption on how the list is constructed. Therefore, the bound holds for any algorithm where $A_t$ is a $\gamma$-approximation at any time $t$, for potentially different values of $\gamma$ than in our paper.

%% file: Experiments.tex

\section{Experiments}
\label{sec:experiments}

This section is organized as follows. In \cref{sec:approximation ratio}, we validate the approximation ratio of the greedy algorithm from \cref{sec:optimal list}. In \cref{sec:baselines}, we introduce our baselines and topic coverage function. A synthetic experiment, which highlights the advantages of our method, is presented in \cref{sec:synthetic}. In \cref{sec:experimental setting}, we describe our experimental setting for the real-world datasets. We evaluate our algorithm on three real-world problems in the rest of the sections. 

\subsection{Empirical Validation of the  Approximation Ratio}
\label{sec:approximation ratio}

\begin{table}
    \centering
    \caption{Approximation ratio of the greedy maximization (\cref{sec:optimal list}) for MovieLens $1$M dataset achieved by exhaustive search for different values of $K$.}
    \begin{tabular}{|c|c|}\hline
      $K$ & Approximation ratio \\ \hline
        1 & 1.0000 \\
        2 & 0.9926 \\
        3 & 0.9997 \\
        4 & 0.9986 \\ \hline
      \end{tabular}
      \label{tab:approximation ratio}
\end{table}

In \cref{sec:optimal list}, we showed that a near-optimal list can be computed greedily. The approximation ratio of the greedy algorithm is close to $1 - 1 / e$ when the maximum click probability is small. Now, we demonstrate empirically that the approximation ratio is close to one in a domain of our interest.

We experiment with MovieLens $1$M dataset from \cref{ssec:movie} (described later). The topic coverage and user preferences are set as in \cref{sec:experimental setting} (described later). We choose $100$ random users and items and vary the number of recommended items $K$ from $1$ to $4$. For each user and $K$, we compute the optimal list $A^\ast$ in \eqref{eq:optimal list} by exhaustive search. Let the corresponding greedy list, which is computed as in \eqref{eq:greedy maximization}, be $A^{\text{greedy}}$. Then $f(A^{\text{greedy}}, \theta^\ast) / f(A^\ast, \theta^\ast)$ is the approximation ratio under user preferences $\theta^\ast$.

For each $K$, we average approximation ratios over all users and report the averages in \cref{tab:approximation ratio}. The average approximation ratio is always more than $0.99$, which means that the greedy maximization in \eqref{eq:greedy maximization} is near optimal. We believe that this is due to the diminishing character of our objective (\cref{sec:optimal list}). The average approximation ratio is $1$ when $K = 1$. This is expected since the optimal list of length $1$ is the most attractive item under $\theta^\ast$, which is always chosen in the first step of the greedy maximization in \eqref{eq:greedy maximization}.
    
\subsection{Baselines and Topic Coverage}
\label{sec:baselines}

We compare $\cascadelsb$ to three baselines. The first baseline is $\lsbgreedy$ \cite{guestrin2011nips}. $\lsbgreedy$ captures diversity but differs from $\cascadelsb$ by assuming feedback at all positions. The second baseline is $\cascadelinucb$ \cite{zong16cascading}, an algorithm for cascading bandits with a linear generalization across items \cite{wen2015efficient,abbasi-yadkori11improved}. To make it comparable to $\cascadelsb$, we set the feature vector of item $e$ as $x_e = \Delta(e \mid \emptyset)$. This guarantees that $\cascadelinucb$ operates in the same feature space as $\cascadelsb$; except that it does not model interactions due to higher ranked items, which lead to diversity. The third baseline is $\cascadeklucb$ \cite{kveton15cascading}, a near-optimal algorithm for cascading bandits that learns the attraction probability of each item independently. This algorithm is expected to perform poorly when the number of items is large. Also, it does not model diversity.

All compared algorithms are evaluated by the $n$-step regret $R^{\gamma}(n)$ with $\gamma = 1$, as defined in \eqref{analysis:cumulative_regret}. We approximate the optimal solution $A^\ast$ by the greedy algorithm in \eqref{eq:greedy maximization}. The learning rate $\sigma$ in $\cascadelsb$ is set to $0.1$. The other parameter $\alpha$ is set to the lowest permissible value, according to \eqref{analysis:c}. The corresponding parameter $\sigma$ in $\lsbgreedy$ and $\cascadelinucb$ is also set to $0.1$. All remaining parameters in the algorithms are set as suggested by their theoretical analyses. $\cascadeklucb$ does not have any tunable parameter.

The topic coverage in \cref{sec:yue model} can be defined in many ways. In this work, we adopt the probabilistic coverage function proposed in El-Arini \etalabbr~\cite{el2009turning},
\begin{align}
  c(S) = \left(1 - \prod_{e \in S} (1 - \bar{w}(e, 1)), \ \dots, \ 1 - \prod_{e \in S} (1 - \bar{w}(e, d))\right)\,,
  \label{eq:topicCoverage}
\end{align}
where $\bar{w}(e, j) \in [0, 1]$ is the \emph{attractiveness} of item $e \in E$ in topic $j \in [d]$. Under the assumption that items cover topics independently, the $j$-th entry of $c(S)$ is the probability that at least one item in $S$ covers topic $j$. Clearly, the proposed function in~\eqref{eq:topicCoverage} is monotone and submodular in each entry of $c(S)$, as required in \cref{sec:yue model}.

\subsection{Synthetic Experiment}
\label{sec:synthetic}

The goal of this experiment is to illustrate the need for modeling both diversity and partial-click feedback. We consider a problem with $L = 53$ items and $d = 3$ topics. We recommend $K = 2$ items and simulate a single user whose preferences are $\theta^\ast = (0.6, 0.4, 0.0)$. The attractiveness of items $1$ and $2$ in topic $1$ is $0.5$, and $0$ in all other topics. The attractiveness of item $3$ in topic $2$ is $0.5$, and $0$ in all other topics. The remaining $50$ items do not belong to any preferred topic of the user. Their attractiveness in topic $3$ is $1$, and $0$ in all other topics. These items are added to make the learning problem harder, as well as to model a real-world scenario where most items are likely to be unattractive to any given user.

The optimal recommended list is $A^\ast = (1, 3)$. This example is constructed so that the optimal list contains only one item from the most preferred topic, either item $1$ or $2$. The $n$-step regret of all the algorithms is shown in \cref{fig:simulated}. We observe several trends.

First, the regret of $\cascadelsb$ flattens and does not increase with the number of steps $n$. This means that $\cascadelsb$ learns the optimal solution.

Second, the regret of $\lsbgreedy$ grows linearly with the number of steps $n$, which means $\lsbgreedy$ does not learn the optimal solution. This phenomenon can be explained as follows. When $\lsbgreedy$ recommends $A^\ast = (1, 3)$, it severely underestimates the preference for topic $2$ of item $3$, because it assumes feedback at the second position even if the first position is clicked. Because of this, $\lsbgreedy$ switches to recommending item $2$ at the second position at some point in time. This is suboptimal. After some time, $\lsbgreedy$ swiches back to recommending item $3$, and then it oscillates between items $2$ and $3$. Therefore, $\lsbgreedy$ has a linear regret and performs poorly.

Third, the regret of $\cascadelinucb$ is linear because it converges to list $(1, 2)$. The items in this list belong to a single topic, and therefore are redundant in the sense that a higher click probability can be achieved by recommending a more diverse list $A^\ast = (1, 3)$.

Finally, the regret of $\cascadeklucb$ also flattens, which means that the algorithm learns the optimal solution. However, because $\cascadeklucb$ does not generalize across items, it learns $A^\ast$ with an order of magnitude higher regret than $\cascadelsb$.

\begin{figure}[t]
  \centering
  \includegraphics[width = 8cm]{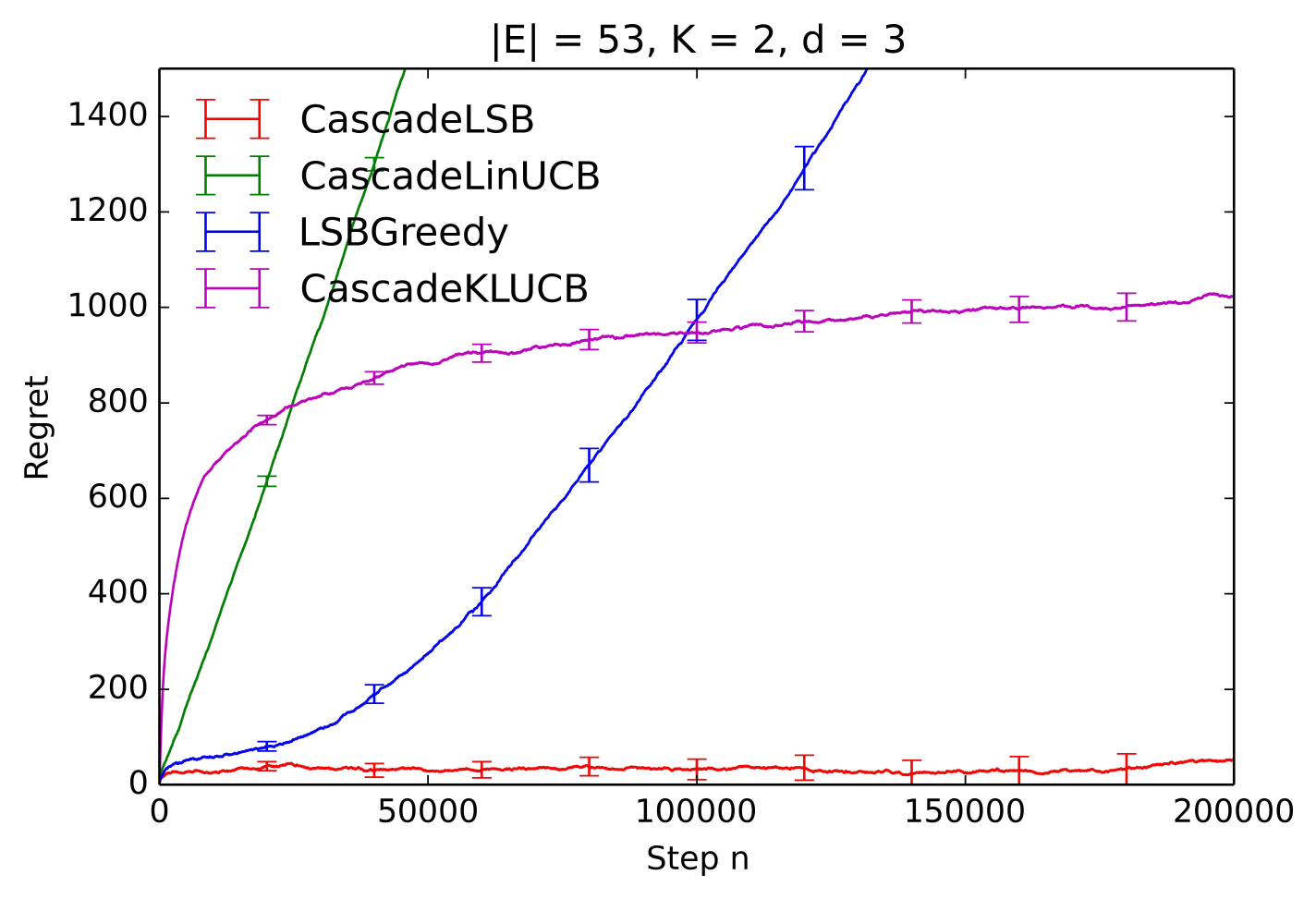}
  \caption{Evaluation on synthetic problem. \cascadelsb's regret is the least and sublinear. \lsbgreedy\ penalizes lower ranked items, thus oscillates between two lists making its regret linear. \cascadelinucb\'s regret is linear because it does not diversify. {{\tt CascadeKL$\mhyphen$UCB}}'s regret is sublinear, however, order of magnitude higher than \cascadelsb's regret.}
  \label{fig:simulated}
  \vspace{-0.5cm}
\end{figure}

\subsection{Real-World Datasets}
\label{sec:experimental setting}

\begin{figure*}[t]
  \centering
  \includegraphics[width=16cm, height=3.6cm]{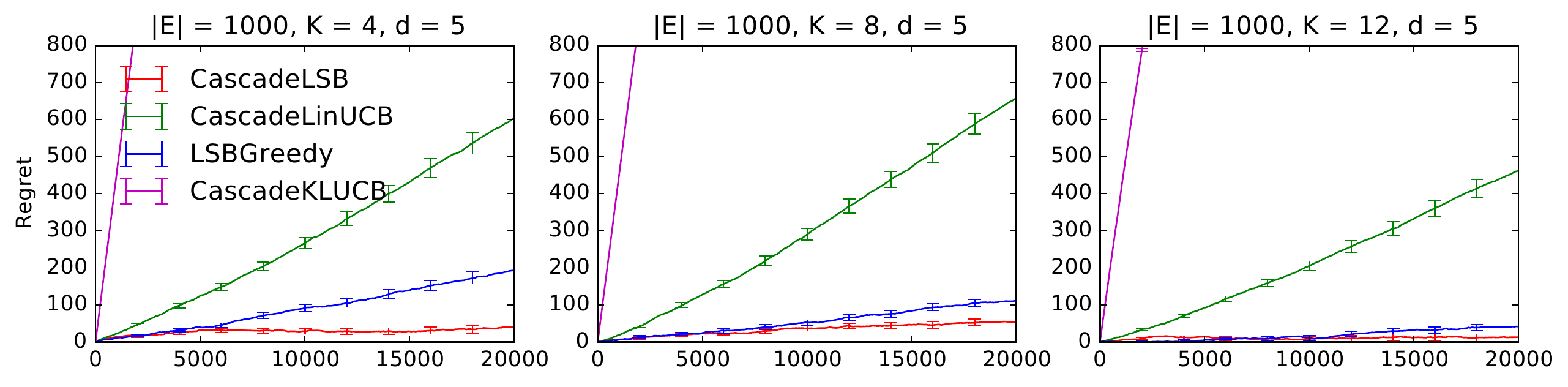}
  \includegraphics[width=16cm, height=3.6cm]{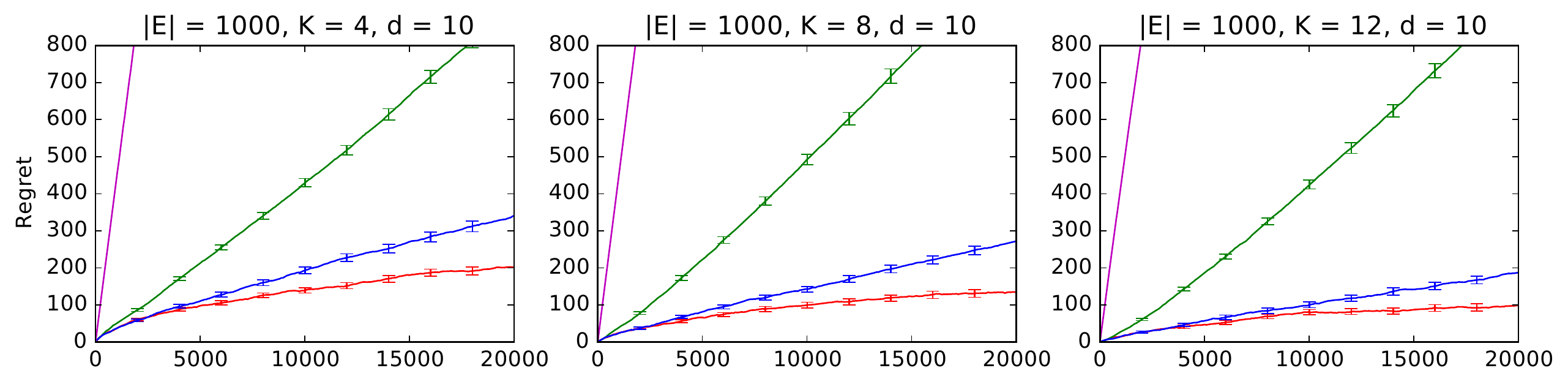}
  \includegraphics[width=16cm, height=3.6cm]{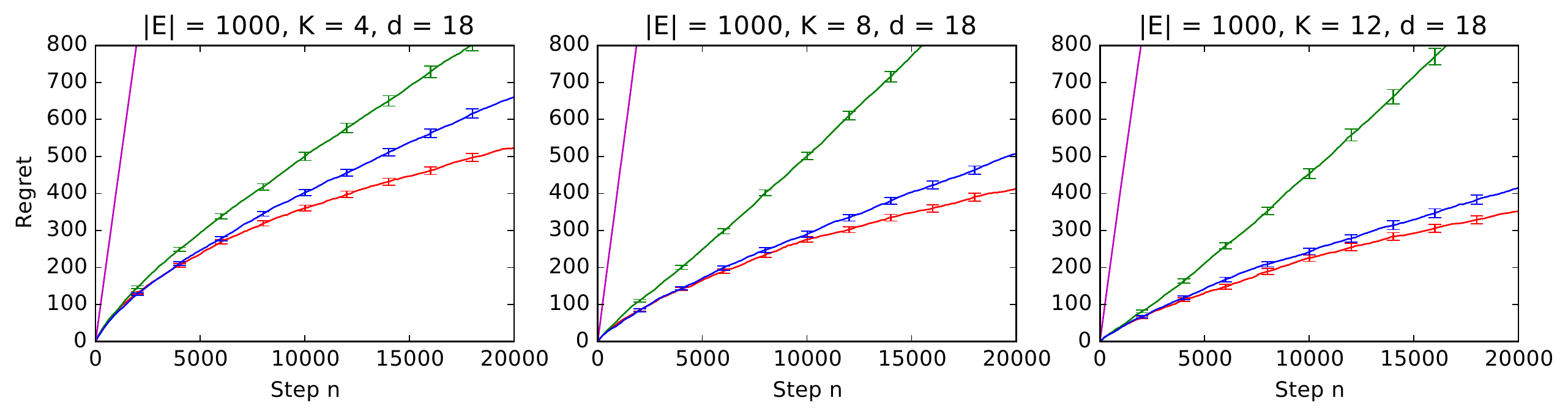}
  \caption{Evaluation on the MovieLens dataset. The number of topics $d$ varies with rows. The number of recommended items $K$ varies with columns. Lower regret means better performance, and sublinear curve represents learning the optimal list. 
  \cascadelsb\, is robust to both number of topics $d$ and size $K$ of the recommended list.}
  \label{fig:MLfinal}
\end{figure*}

Now, we assess \cascadelsb\, on real-world datasets. One approach to evaluating bandit policies without a live experiment is \emph{off-policy evaluation} \cite{li2011unbiased}. Unfortunately, off-policy evaluation is unsuitable for our problem because the number of actions, all possible lists, is exponential in $K$. Thus, we evaluate our policies by building an interaction model of users from past data. This is another popular approach and is adopted by most papers discussed in \cref{sec:related work}.

All of our real-world problems are associated with a set of users $U$, a set of items $E$, and a set of topics $[d]$. The relations between the users and items are captured by feedback matrix $F \in \set{0, 1}^{|U| \times |E|}$, where row $u$ corresponds to user $u \in U$, column $i$ corresponds to item $i \in  E$, and $F(u, i)$ indicates if user $u$ was attracted to item $i$ in the past. The relations between items and topics are captured by matrix $G \in \set{0, 1}^{|E| \times d}$, where row $i$ corresponds to item $i \in  E$, column $j$ corresponds to topic $j \in  [d]$, and $G(i, j)$ indicates if item $i$ belongs to topic $j$. Next, we describe how we build these matrices.

The \emph{attraction probability of item $i$ in topic $j$} is defined as the number of users who are attracted to item $i$ over all users who are attracted to at least one item in topic $j$. Formally,
\begin{align}
  \textstyle
  \bar{w}(i, j) =
  \sum\limits_{u \in U} F(u, i) G(i, j)
  \left[\sum\limits_{u \in U} \I{\exists i' \in E : F(u, i') G(i', j) > 0}\right]^{-1}\,.
  \label{eq:define_w}
\end{align}
Therefore, the attraction probability represents a relative worth of item $i$ in topic $j$. We illustrate this concept with the following example. Suppose that the item is popular, such as movie \emph{Star Wars} in topic \emph{Sci-Fi}. Then \emph{Star Wars} attracts many users who are attracted to at least one movie in topic \emph{Sci-Fi}, and its attraction probability in topic \emph{Sci-Fi} should be close to one.

The \emph{preference of a user $u$ for topic $j$} is the number of items in topic $j$ that attracted user $u$ over the total number of topics of all items that attracted user $u$, i.e., 
\begin{align}
  \textstyle
  \theta^\ast_j =
  \sum\limits_{i \in E} F(u, i) G(i, j)
  \left[\sum\limits_{j' \in [d]}\sum\limits_{i \in E} F(u, i) G(i, j')\right]^{-1}\,.
  \label{eq:define_theta}
\end{align}
Note that $\sum_{j = 1}^d \theta^\ast_j = 1$. Therefore, $\theta^\ast = (\theta^\ast_1, \dots, \theta^\ast_d)$ is a probability distribution over topics for user $u$.

We divide users randomly into two halves to form training and test sets. This means that the feedback matrix $F$ is divided into two matrices, $F_\text{train}$ and $F_\text{test}$. The parameters that define our click model, which are computed from $\bar{w}(i, j)$ in \eqref{eq:topicCoverage} and $\theta^\ast$ in \eqref{eq:define_theta}, are estimated from $F_\text{test}$ and $G$. The topic coverage features in $\cascadelsb$, which are computed from $\bar{w}(i, j)$ in \eqref{eq:topicCoverage}, are estimated from $F_\text{train}$ and $G$. This split ensures that the learning algorithm does not have access to the optimal features for estimating user preferences, which is likely to happen in practice. In all experiments, our goal is to maximize the probability of recommending at least one attractive item. The experiments are conducted for $n = 20$k steps and averaged over $100$ random problem instances, each of which corresponds to a randomly chosen user.

\subsection{Movie Recommendation}
\label{ssec:movie}

\begin{figure*}[t]
  \centering
  \includegraphics[width=16cm, height = 3.6cm]{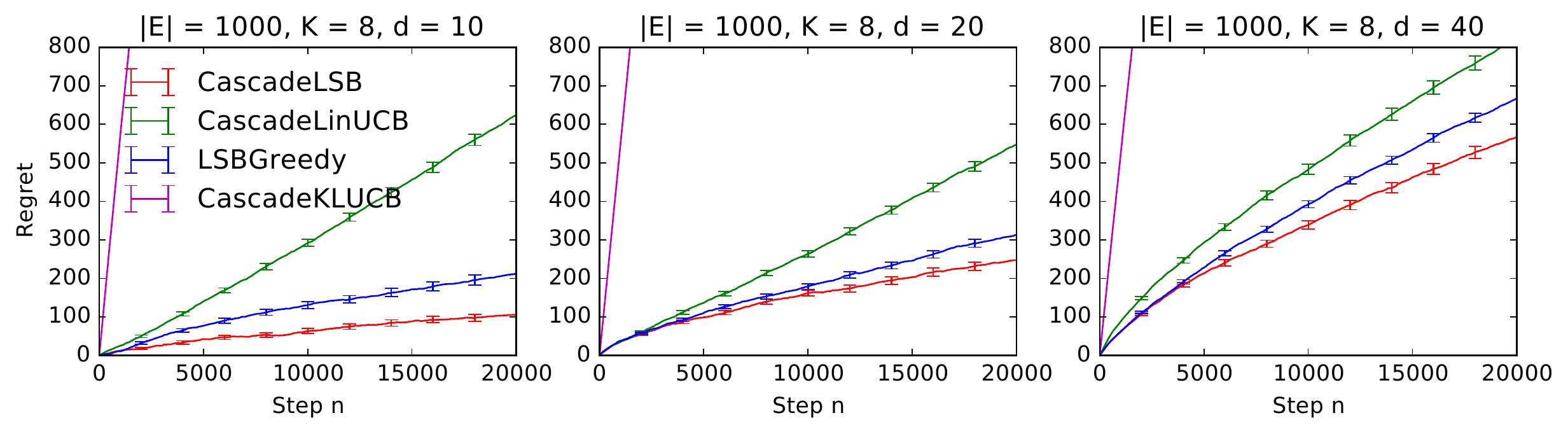}
  \vspace{-0.1in}
  \caption{Evaluation on the Million Song dataset. $K = 8$ and vary $d$ from $10$ to $40$. Lower regret means better performance, and sublinear curve represents learning the optimal list.
 \cascadelsb\ is shown to be robust to number of topics $d$.}
  \label{fig:MS_k8}
\end{figure*}

Our first real-world experiment is in the domain of movie recommendations. We experiment with the MovieLens $1$M dataset\footnote{http://grouplens.org/datasets/movielens/}. The dataset contains $1$M ratings of $4$k movies by $6$k users who joined MovieLens in the year $2000$. We extract $|E| = 1000$ most rated movies and $|U| = 1000$ most rating users. These active users and items are extracted just for simplicity and to have confident estimates of $\bar{w}(.)$ \eqref{eq:define_w} and $\theta^\ast$ \eqref{eq:define_theta} for experiments.

We treat movies and their genres as items and topics, respectively. The ratings are on a $5$-star scale. We assume that user $i$ is attracted to movie $j$ if the user rated that movie with $5$ stars, $F(i, j) = \I{\text{user $i$ rated movie $j$ with $5$ stars}}$. For this definition of attraction, about $8\%$ of user-item pairs in our dataset are attractive. We assume that a movie belongs to a genre if it is tagged with that particular genre. In this experiment, we vary the number of topics $d$ from $5$ to $18$ (maximum possible genres in MovieLens 1M dataset), as well as the number of recommended items $K$ from $4$ to $12$. The topics are sorted in the descending order of the number of items in them. While varying topics, we choose the most popular ones.

Our results are reported in \cref{fig:MLfinal}. We observe that $\cascadelsb$ has the lowest regret among all compared algorithms for all $d$ and $K$. This suggests that $\cascadelsb$ is robust to choice of both the parameters $d$ and $K$. For $d = 18$, $\cascadelsb$ achieves almost $20\%$ lower regret than the best performing baseline, $\lsbgreedy$. $\lsbgreedy$ has a higher regret than $\cascadelsb$ because it learns from unexamined items. $\cascadeklucb$ performs the worst because it learns one attraction weight per item. This is impractical when the number of items is large, as in this experiment. The regret of $\cascadelinucb$ is linear, which means that it does not learn the optimal solution. This shows that linear generalization in the cascade model is not sufficient to capture diversity. More sophisticated models of user interaction, such as the diverse cascade model (\cref{sec:model}), are necessary. At $d = 5$, the regret of $\cascadelsb$ is low. As the number of topics increase, the problems become harder and the regret of $\cascadelsb$ increases.

\subsection{Million Song Recommendation}
\label{ssec:millionsong}

Our next experiment is in song recommendation domain. We experiment with the Million Song dataset\footnote{http://labrosa.ee.columbia.edu/millionsong/}, which is a collection of audio features and metadata for one million pop songs. We extract $|E| = 1000$ most popular songs and $|U| = 1000$ most active users, as measured by the number of song-listening events. These active users and items provide more confident estimates of $\bar{w}(.)$ \eqref{eq:define_w} and $\theta^\ast$ \eqref{eq:define_theta} for experiments.

We treat songs and their genres as items and topics, respectively. We assume that a user $i$ is attracted to a song $j$ if the user had listened to that song at least $5$ times. Formally this is captured as $F(i, j) = \I{\text{user $i$ had listened to song $j$ at least $5$ times}}$. By this definition, about $3\%$ of user-item pairs in our dataset are attractive. We assume that a song belongs to a genre if it is tagged with that genre. Here, we fix the number of recommended items at $K = 8$ and vary the number of topics $d$ from $10$ to $40$.

Our results are reported in \cref{fig:MS_k8}. Again, we observe that $\cascadelsb$ has the lowest regret among all compared algorithms. This happens for all $d$, and we conclude that $\cascadelsb$ is robust to the choice of $d$. At $d = 40$, $\cascadelsb$ has about $15\%$ lower regret than the best performing baseline, $\lsbgreedy$. Compared to the previous experiment, $\cascadelinucb$ learns a better solution over time at $d = 40$. However, it still has about $50\%$ higher regret than $\cascadelsb$ at $n = 20$k. Again, $\cascadeklucb$ performs the worst in all the experiments.

\subsection{Restaurant Recommendation}
\label{ssec:restuarant}

\begin{figure*}[t]
  \centering
  \includegraphics[width=16cm, height=3.6cm]{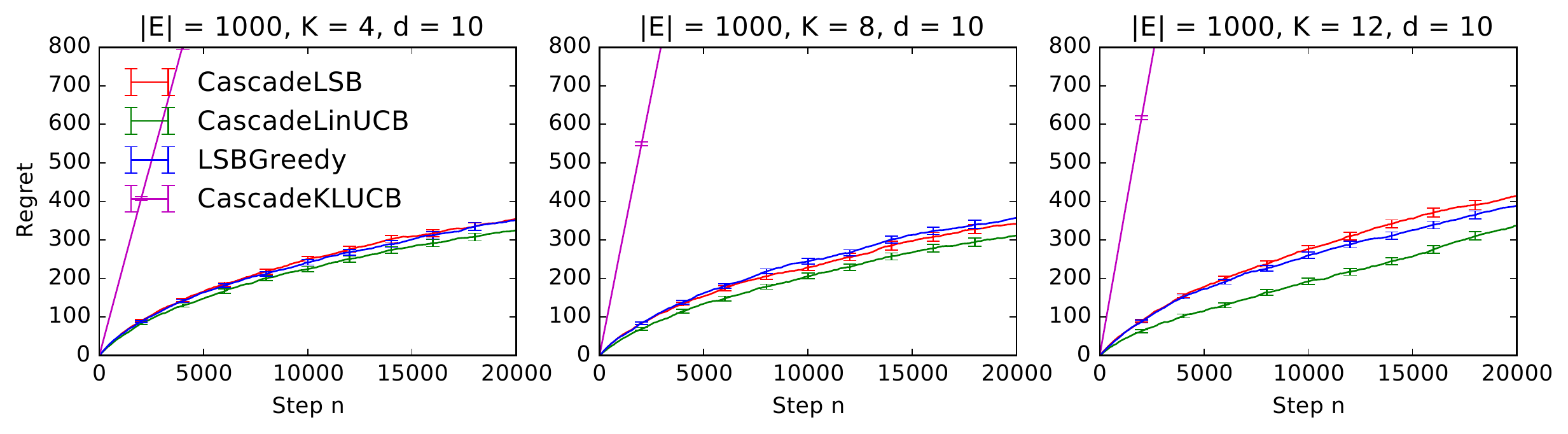}
  \vspace{-0.1in}
  \caption{Evaluation on the Yelp dataset. We fix $d = 10$ and vary $K$ from $4$ to $12$. Lower regret means better performance, and sublinear curve represents learning the optimal list. All the algorithms except $\cascadeklucb$ perform similar due to small attraction probabilities of items.}
  \label{fig:Yelp_k8}
\end{figure*}

Our last real-world experiment is in the domain of restaurant recommendation. We experiment with the Yelp Challenge dataset\footnote{https://www.yelp.com/dataset\_challenge, Round 9 Dataset}, which contains $4.1$M reviews written by $1$M users for $48$k restaurants in more than $600$ categories. Consistent with the above experiments, here again, we extract $|E| = 1000$ most reviewed restaurants and $|U| = 1000$ most reviewing users. 

We treat restaurants and their categories as items and topics, respectively. We assume that user $i$ is attracted to restaurant $j$ if the user rated that restaurant with at least $4$ stars, i.e., $F(i, j) = \I{\text{user $i$ rated restaurant $j$ with at least $4$ stars}}$. For this definition of attraction, about $3\%$ of user-item pairs in our dataset are attractive. We assume that a restaurant belongs to a category if it is tagged with that particular category. In this experiment, we fix the number of topics at $d = 10$ and vary the number of recommended items $K$ from $4$ to $12$.

Our results are reported in \cref{fig:Yelp_k8}. Unlike in the previous experiments, we observe that $\cascadelinucb$ performs comparably to $\cascadelsb$ and $\lsbgreedy$. We investigated this trend and discovered that this is because the attraction probabilities of items, as defined in \eqref{eq:define_w}, are often very small, such as on the order of $10^{-2}$. This means that the items do not cover any topic properly. In this setting, the gain in topic coverage due to any item $e$ over higher ranked items $S$, $\Delta(e \mid S)$, is comparable to $\Delta(e \mid \emptyset)$ when $\abs{S}$ is small. This follows from the definition of the coverage function in \eqref{eq:topicCoverage}. Now note that the former are the features in $\cascadelsb$ and $\lsbgreedy$, and the latter are the features in $\cascadelinucb$. Since the features are similar, all algorithms choose lists with respect to similar objectives, and their solutions are similar. 


%% file: RelatedWork.tex

\section{Related Work}
\label{sec:related work}

Our work is closely related to two lines of work in online learning to rank, in the cascade model \cite{kveton15cascading,combes15learning} and with diversity \cite{radlinski08learning,guestrin2011nips}.

\emph{Cascading bandits} \cite{kveton15cascading,combes15learning} are an online learning model for learning to rank in the cascade model \cite{craswell08experimental}. Kveton \etalabbr~\cite{kveton15cascading} proposed a near-optimal algorithm for this problem, $\cascadeklucb$. Several papers extended cascading bandits \cite{kveton15cascading,combes15learning,kveton15combinatorial,katariya16dcm,zong16cascading,li2016contextual}. The most related to our work are linear cascading bandits of Zong \etalabbr~\cite{zong16cascading}. Their proposed algorithm, $\cascadelinucb$, assumes that the attraction probabilities of items are a linear function of the features of items, which are known; and an unknown parameter vector, which is learned. This work does not consider diversity. We compare to $\cascadelinucb$ in \cref{sec:experiments}.

Yue and Guestrin \cite{guestrin2011nips} studied the problem of online learning with diversity, where each item covers a set of topics. They also proposed an efficient algorithm for solving their problem, $\lsbgreedy$. $\lsbgreedy$ assumes that if the item is not clicked, then the item is not attractive and penalizes the topics of this item. We compare to $\lsbgreedy$ in \cref{sec:experiments} and show that its regret can be linear when clicks on lower-ranked items are biased due to clicks on higher-ranked items.

The difference of our work from $\cascadelinucb$ and $\lsbgreedy$ can be summarized as follows. In \cascadelinucb, the click on the $k$-th recommended item depends on the features of that item and clicks on higher-ranked items. In \lsbgreedy, the click on the $k$-th item depends on the features of all items up to position $k$. In \cascadelsb, the click on the $k$-th item depend on the features of all items up to position $k$ and clicks on higher-ranked items.

Raman et al.~\cite{raman2012online} also studied the problem of online diverse recommendations. The key idea of this work is preference feedback among rankings. That is, by observing clicks on the recommended list, one can construct a ranked list that is more preferred than the one presented to the user. However, this approach assumes full feedback on the entire ranked list. When we consider partial feedback, such as in the cascade model, a model based on comparisons among rankings cannot be applied. Therefore, we do not compare to this approach.

\emph{Ranked bandits} are a popular approach to online learning to rank \cite{radlinski08learning,slivkins13ranked}. The key idea in ranked bandits is to model each position in the recommended list as a separate bandit problem, which is solved by a \emph{base bandit algorithm}. The optimal item at the first position is clicked by most users, the optimal item at the second position is clicked by most users that do not click on the first optimal item, and so on. This list is diverse in the sense that each item in this list is clicked by many users that do not click on higher-ranked items. This notion of diversity is different from our paper, where we learn a list of diverse items over topics for a single user. Learning algorithms for ranked bandits perform poorly in cascade-like models \cite{kveton15cascading,katariya16dcm} because they learn from clicks at all positions. We expect similar performance of other learning algorithms that make similar assumptions \cite{kohli13fast}, and thus do not compare to them.

None of the aforementioned works considers all three aspects of learning to rank that are studied in this paper: \emph{online}, \emph{diversity}, and \emph{partial-click feedback}.

%% file: Conclusions.tex

\section{Conclusions and Future Work}
\label{sec:conclusions}

Diverse recommendations address the problem of ambiguous user intent and reduce redundancy in recommended items \cite{radlinski2009redundancy,adomavicius2012improving}. In this work, we propose the first online algorithm for learning to rank diverse items in the cascade model of user behavior \cite{craswell08experimental}. Our algorithm is computationally efficient and easy to implement. We derive a gap-free upper bound on its scaled $n$-step regret, which is sublinear in the number of steps $n$. We evaluate the algorithm in several synthetic and real-world experiments, and show that it is competitive with a range of baselines.

To show that we learn the optimal preference vector $\theta^\ast$, we assumed $\theta^\ast$ is fixed. In reality, it might change over time maybe due to inherent changes in user preferences or due to serendipitous recommendations themselves provided to the user~\cite{ziegler2005improving}. However, we emphasize that all online learning to rank algorithms \cite{radlinski08learning,slivkins13ranked} including ours can handle such changes in preferences over time. 

One limitation of our work is that we assume that the user clicks on at most one recommended item. We would like to stress that this assumption is only for simplicity of exposition. In particular, the algorithm of Katariya \etalabbr~\cite{katariya16dcm}, which learns to rank from multiple clicks, is almost the same as $\cascadeklucb$ \cite{kveton15cascading}, which learns to rank from at most one click in the cascade model. The only difference is that Katariya \etalabbr~\cite{katariya16dcm} consider feedback up to the last click. We believe that $\cascadelsb$ can be generalized in the same way, by considering all feedback up to the last click. We leave this for future work.

%% file: proof.tex

\clearpage
\onecolumn
\appendix

\section{Proof for Approximation Ratio}
\label{sec:approximation_ratio}

We prove Theorem~\ref{thm:approximation_ratio} in this section. First, we prove the following technical lemma:
\begin{lemma}
\label{lemma:lower_bounds}
For any positive integer $K=1,2,\ldots$, and any real numbers $b_1, \ldots, b_K \in [0, B]$, where $B$ is a real number in $[0,1]$, we have the following bounds
\[
\max \left \{\frac{1}{K}, 1- \frac{K-1}{2} B  \right \} \sum_{k=1}^K b_k \leq 
1- \prod_{k=1}^K (1- b_k) \leq \sum_{k=1}^K b_k .
\]
\end{lemma}
\begin{proof}
First, We prove that $1- \prod_{k=1}^K (1- b_k) \leq \sum_{k=1}^K b_k$ by induction. 
Notice that when $K=1$, this inequality trivially holds. Assume that this inequality holds for $K$, we prove that it also holds for $K+1$. Note that
\begin{align}
1- \prod_{k=1}^{K+1} (1- b_k) =& \, \left[1- \prod_{k=1}^{K} (1- b_k) \right] [1- b_{K+1}] + b_{K+1} \nonumber \\
\stackrel{(a)}{\leq}& \, \left[ \sum_{k=1}^K b_k \right] [1- b_{K+1}] + b_{K+1} \nonumber \\
\leq & \, \sum_{k=1}^{K+1} b_k \, ,
\end{align}
where (a) follows from the induction hypothesis. This concludes the proof for the upper bound $1- \prod_{k=1}^K (1- b_k) \leq \sum_{k=1}^K b_k$.

Second, we prove that $1- \prod_{k=1}^K (1- b_k) \geq \frac{1}{K} \sum_{k=1}^K b_k $. Notice that this trivially follows from the fact that
\begin{align}
1- \prod_{k=1}^K (1- b_k) \geq \max_{k} b_k \geq \frac{1}{K} \sum_{k=1}^K b_k \, .
\end{align}

Finally, we prove the lower bound $1- \prod_{k=1}^K (1- b_k) \geq \left[ 1- \frac{K-1}{2} B \right] \sum_{k=1}^K b_k$ by induction.\\

\noindent \textbf{Base Case:} Notice that when $K=1$, we have
\[
1- \prod_{k=1}^K (1- b_k)  = b_1 = \left[ 1- \frac{K-1}{2} B \right]\sum_{k=1}^K b_k  .
\]
That is, the lower bound trivially holds for the case with $K=1$.\\

\noindent \textbf{Induction:} Assume that the lower bound holds for $K$, we prove that it also holds for $K+1$. Notice that if $1-\frac{K}{2}B \leq 0$, then this lower bound holds trivially. For the non-trivial case with $1-\frac{K}{2}B > 0$, we have
\begin{align}
1- \prod_{k=1}^{K+1} (1- b_k) = & \, \frac{1}{K+1} \sum_{i=1}^{K+1}  \left \{ (1- b_i) \left [ 
1 - \prod_{k \neq i} (1- b_k)
\right] + b_i \right \} \nonumber \\
\stackrel{(a)}{\geq } & \,
\frac{1}{K+1} \sum_{i=1}^{K+1}  \left \{ (1- b_i) 
\left[ 1- \frac{K-1}{2} B \right]\sum_{k \neq i} b_k
 + b_i \right \}  \nonumber \\
 \stackrel{(b)}{\geq } & \,
\frac{1}{K+1} \sum_{i=1}^{K+1}  \left \{ (1- B) 
\left[ 1- \frac{K-1}{2} B \right]\sum_{k \neq i} b_k
 + b_i \right \}  \nonumber \\
 \stackrel{(c)}{= } & \,
 \frac{K}{K+1}  \left \{ (1- B) 
\left[ 1- \frac{K-1}{2} B \right]\sum_{k =1}^{K+1} b_k \right \}  
 + \frac{1}{K+1} \sum_{k =1}^{K+1} b_k \nonumber \\
 = & \, \left \{ 
 1- \frac{K}{2}B + \frac{K(K-1)}{2(K+1)} B^2
 \right \} \sum_{k =1}^{K+1} b_k \geq \left \{ 
 1- \frac{K}{2}B 
 \right \} \sum_{k =1}^{K+1} b_k \, ,
\end{align}
where (a) follows from the induction hypothesis, (b) follows from the fact that $b_i \leq B$ for all $i$ and $1- \frac{K-1}{2} B >0$, and (c) follows from the fact that
$\sum_{i=1}^{K+1} \sum_{k \neq i} b_k = K \sum_{k=1}^{K+1} b_k $. This concludes the proof.
\end{proof}

\noindent
We have the following remarks on the results of Lemma~\ref{lemma:lower_bounds}:
\begin{remark}
Notice that the lower bound $1- \prod_{k=1}^K (1- b_k) \geq \frac{1}{K} \sum_{k=1}^K b_k $ is tight when $b_1=b_2=\ldots=b_K=1$. So we cannot further improve this lower bound without imposing additional constraints on $b_k$'s.
\end{remark}
\begin{remark}
From Lemma~\ref{lemma:lower_bounds}, we have
\[
1- \frac{K-1}{2} B \leq \frac{1- \prod_{k=1}^K (1- b_k)}{\sum_{k=1}^K b_k} \leq 1 \, .
\]
Thus, if $B(K-1) \ll 1$, then $1- \prod_{k=1}^K (1- b_k) \approx \sum_{k=1}^K b_k$. Moreover, for any fixed $K$, we have
$\lim_{B \downarrow 0} \frac{1- \prod_{k=1}^K (1- b_k)}{\sum_{k=1}^K b_k} =1$.
\end{remark}

\noindent
We now prove Theorem~\ref{thm:approximation_ratio} based on Lemma~\ref{lemma:lower_bounds}. Notice that by definition of $c_{\max}$, we have $\langle\Delta(a_k \mid \set{a_1, \dots, a_{k - 1}}), \theta^\ast\rangle \leq c_{\max}$. From Lemma~\ref{lemma:lower_bounds}, for any $A \in \Pi_K(E)$, we have
\begin{equation}
\label{eq:approximation_ratio_aux_1}
\max \left \{\frac{1}{K}, 1- \frac{K-1}{2} c_{\max}  \right \} \langle c(A), \theta^\ast\rangle \leq f(A, \theta^\ast) \leq \langle c(A), \theta^\ast\rangle \, .
\end{equation}
Consequently, we have
\begin{align}
f(A^{\text{greedy}}, \theta^{\ast}) \stackrel{(a)}{\geq}  &  \,
 \max \left \{\frac{1}{K}, 1- \frac{K-1}{2} c_{\max}  \right \} \langle c(A^{\text{greedy}}), \theta^\ast \rangle  \nonumber \\
 \stackrel{(b)}{\geq} & \, (1-e^{-1})\max \left \{\frac{1}{K}, 1- \frac{K-1}{2} c_{\max}  \right \}
 \max_{A \in \Pi_K(E)} \langle c(A), \theta^\ast \rangle \nonumber \\
\stackrel{(c)}{\geq} & \, (1-e^{-1})\max \left \{\frac{1}{K}, 1- \frac{K-1}{2} c_{\max}  \right \}
  \langle c(A^\ast), \theta^\ast \rangle \nonumber \\
 \stackrel{(d)}{\geq} & \, (1-e^{-1})\max \left \{\frac{1}{K}, 1- \frac{K-1}{2} c_{\max}  \right \}
 f(A^\ast, \theta^{\ast}) \, , 
\end{align}
where (a) and (d) follow from (\ref{eq:approximation_ratio_aux_1}); (b) follows from the facts that $\langle c(A), \theta^\ast \rangle$ is a monotone and submodular set function in $A$ and $A^{\text{greedy}}$ is computed based on the greedy algorithm; and (c) trivially follows from the fact that $A^\ast \in \Pi_K(E)$. This concludes the proof for Theorem~\ref{thm:approximation_ratio}.

\section{Proof for Regret Bound}
\label{sec:proof}


We start by defining some useful notations. Let $\Pi(E)=\bigcup_{k=1}^L \Pi_k(E)$ be the set of all (ordered) 
lists of set $E$ with cardinality $1$ to $L$,
and $w: \Pi(E) \rightarrow [0,1]$ 
be an arbitrary weight function for lists. For any $A \in \Pi(E)$ and any $w$, we define
\begin{equation}
\label{eq:notation_h}
\textstyle h(A, w) = 1- \prod_{k=1}^{|A|} \left[ 1- w(A^k) \right],
\end{equation}
where $A^k$ is the prefix of $A$ with length $k$. With a little bit abuse of notation, we also define the feature $\Delta(A)$ for list
$A=(a_1, \ldots, a_{|A|})$ as $\Delta(A)=\Delta ( a_{|A|}  | \{a_1, \ldots, a_{|A|-1} \} )$. Then, we define the weight function
$\bar{w}$, its high-probability upper bound $U_t$, and its high-probability lower bound $L_t$ as
\begin{align}
\label{eq:measures}
  \bar{w}(A) & = \Delta(A)^T \theta^\ast\,, \nonumber \\
  U_t(A) & = \mathrm{Proj}_{[0,1]} \left[  \Delta(A) ^T \bar{\theta}_t + \alpha \sqrt{\Delta(A)^T M_t^{-1} \Delta(A)} \right]  \,, \nonumber \\
  L_t(A) & = \mathrm{Proj}_{[0,1]} \left[ \Delta(A) ^T \bar{\theta}_t - \alpha \sqrt{\Delta(A)^T M_t^{-1} \Delta(A)} \right ]
\end{align}
for any ordered list $A$ and any time $t$.
Note that $\mathrm{Proj}_{[0,1]} \left[ \cdot \right]$ projects a real number onto interval $[0,1]$,
and based on Equation~\ref{eq:click probability}, \ref{eq:notation_h}, 
and \ref{eq:measures},
 we have
$h(A, \bar{w})=f(A, \theta^*)$ for all ordered list $A$. We also use $\cH_t$ to denote the history 
of past actions and observations
by the end of time period $t$. Note that $U_{t-1}$, $L_{t-1}$ and $A_t$ are all deterministic 
conditioned on $\cH_{t-1}$.
For all time $t$, we define the ``good event" as 
$\cE_t = \left \{
 L_t(A) \leq \bar{w}(A) \leq U_t(A), \, \forall A \in \Pi(E)
\right \}$,
and $\bar{\cE}_t$ as the complement of $\cE_t$. Notice that both $\cE_{t-1}$ and $\bar{\cE}_{t-1}$ are also deterministic 
conditioned on $\cH_{t-1}$.
Hence, we have
\begin{align}
& \, \E{ f(A^\ast, \theta^\ast) - f(A_t, \theta^\ast)/\gamma}{} = \E{ h(A^\ast, \bar{w}) - h(A_t, \bar{w})/\gamma}{} \nonumber \\
\leq & \, P(\cE_{t-1})  \E{ h(A^*, \bar{w}) - h(A_t, \bar{w})/\gamma \middle | \cE_{t-1}}{} + P(\bar{\cE}_{t-1}), \label{analysis:partial_1}
\end{align}
where the above inequality follows from the naive bound that $h(A^*, \bar{w}) - h(A_t, \bar{w})/\gamma \leq 1$.
Notice that under event $\cE_{t-1}$, we have $h(A, L_{t-1}) \leq h(A, \bar{w}) \leq h(A, U_{t-1})$ for all ordered list
$A$. Thus, we have $h(A^*, \bar{w}) \leq h(A^*, U_{t-1})$.
On the other hand, since $A_t$ is computed based on a $\gamma$-approximation algorithm, by definition
\[
h(A^*, U_{t-1}) \leq \max_{A \in \Pi_{K}(E)} h(A, U_{t-1}) \leq h(A_t, U_{t-1})/\gamma.
\]
Combining the above inequalities, under event $\cE_{t-1}$, we have
\[
h(A^*, \bar{w}) - h(A_t, \bar{w})/\gamma
\leq \frac{1}{\gamma}  \left[ h(A_t, U_{t-1}) - h(A_t, \bar{w}) \right].
\]
Recall that $A_t^k$ is the prefix of $A_t$ with length $k$, then we have
\begin{align}
& \, h(A_t, U_{t-1}) - h(A_t, \bar{w}) =\prod_{k=1}^K(1- \bar{w}(A_t^k)) -  \prod_{k=1}^K(1-U_{t-1}(A_t^k)) \nonumber \\
=& \, \sum_{k=1}^K \left[ \prod_{i=1}^{k-1} (1- \bar{w}(A_t^i)) \right] \left(U_{t-1}(A_t^k) - \bar{w}(A_t^k) \right)
\left[\prod_{j=k+1}^{K} (1-U_{t-1}(A_t^j) )  \right] \nonumber \\
\leq & \sum_{k=1}^K \left[ \prod_{i=1}^{k-1} (1- \bar{w}(A_t^i)) \right] \left(U_{t-1}(A_t^k) - \bar{w}(A_t^k) \right), \nonumber
\end{align}
where the last inequality follows from the fact that $0 \leq U_{t-1}(A_t^j) \leq 1$.
Let $\cG_{tk}$ be the event that item $a^t_k$ is examined at time $t$, then we have
$\E{\mathbf{1} \left[ \cG_{tk} \right] \middle | \cH_{t-1}}{}=\prod_{i=1}^{k-1} (1- \bar{w}(A_t^i)) $.
Moreover, since $\bar{w}(A_t^k) \geq L_{t-1}(A_t^k)$ under event $\cE_{t-1}$ and 
$\cE_{t-1}$ is deterministic conditioned on $\cH_{t-1}$, for any $\cH_{t-1}$ s.t. $\cE_{t-1}$ holds, we have
\begin{align}
& \, \E{h(A_t, U_{t-1}) - h(A_t, \bar{w}) \middle | \cH_{t-1}}{} \nonumber \\
\leq & \, \textstyle \sum_{k=1}^K \E{\mathbf{1} \left[ \cG_{tk} \right] \middle | \cH_{t-1}}{} 
\left[U_{t-1}(A_t^k) - L_{t-1} (A_t^k)  \right] \nonumber \\
\stackrel{(a)}{\leq} & \, 2\alpha \E{ \textstyle \sum_{k=1}^K \mathbf{1} \left[ \cG_{tk} \right] \sqrt{\Delta(A_t^k)^T M_{t-1}^{-1} \Delta(A_t^k)} \middle | \cH_{t-1}}{} \nonumber \\
\stackrel{(b)}{=} & \, 2\alpha \E{ \textstyle \sum_{k=1}^{\min\{C_t, K\}}  \sqrt{\Delta(A_t^k)^T M_{t-1}^{-1} \Delta(A_t^k)} \middle | \cH_{t-1}}{},\nonumber
\end{align}
where (a) follows from the definitions of $U_{t-1}$ and $L_{t-1}$ (see Equation~\ref{eq:measures}), and
(b) follows from the definitions of $C_t$ and $\cG_{tk}$.
Plug the above inequality into Equation~\ref{analysis:partial_1}, we have
\begin{align}
& \, \E{ f(A^*, \theta^\ast) - f(A_t, \theta^\ast)/\gamma}{} \nonumber \\
\leq & \,  P(\cE_{t-1}) \frac{2\alpha}{\gamma} \E{ \textstyle \sum_{k=1}^{\min\{C_t, K\}}  \sqrt{\Delta(A_t^k)^T M_{t-1}^{-1} \Delta(A_t^k)} \middle | \cE_{t-1}}{} + P(\bar{\cE}_{t-1})
\nonumber \\
\leq & \frac{2\alpha}{\gamma} \E{ \textstyle \sum_{k=1}^{\min\{C_t, K\}}  \sqrt{\Delta(A_t^k)^T M_{t-1}^{-1} \Delta(A_t^k)} }{} + P(\bar{\cE}_{t-1}). \nonumber
\end{align}
So we have
\begin{align}
R^{\gamma}(n) \leq \frac{2\alpha}{\gamma} \E{ \sum_{t=1}^n  \sum_{k=1}^{\min\{C_t, K\}}  \sqrt{\Delta(A_t^k)^T M_{t-1}^{-1} \Delta(A_t^k)} }{} + \sum_{t=1}^n P(\bar{\cE}_{t-1}). \nonumber
\end{align}
The regret bound can be obtained based on a worst-case bound on 
$\sum_{t=1}^n  \sum_{k=1}^{\min\{C_t, K\}}  \sqrt{\Delta(A_t^k)^T M_{t-1}^{-1} \Delta(A_t^k)}$, and a bound on 
$P(\bar{\cE}_{t-1})$. The derivations of these two bounds are the same as in Zong \etalabbr~\cite{zong16cascading}.
Specifically, we have
\begin{lemma}
\label{lemma:worst_case_bound}
The following worst-case bound holds
\[ \sum_{t=1}^n  \sum_{k=1}^{\min\{C_t, K\}}  \sqrt{\Delta(A_t^k)^T M_{t-1}^{-1} \Delta(A_t^k)} \leq K \sqrt{\frac{dn \log \left[ 1 + \frac{nK}{d \sigma^2}\right]}{\log \left( 1+ \frac{1}{\sigma^2}\right)}}.
\]
\end{lemma}
Please refer to Lemma~2 in Zong \etalabbr~\cite{zong16cascading} for the derivation of Lemma~\ref{lemma:worst_case_bound}.
We also have the following bound on $P(\bar{\cE}_{t})$:

\begin{lemma}
\label{lemma:probability_bound} For any $t$, $\sigma > 0$, $\delta \in (0,1)$, and
\[
\alpha \geq \frac{1}{\sigma} \sqrt{d \log \left( 1 + \frac{nK}{d \sigma^2}\right) + 2 \log \left(\frac{1}{\delta} \right)} + \|\theta^\ast \|_2,
\]
we have $P(\bar{\cE}_{t})\leq \delta$.
\end{lemma}
Please refer to Lemma~3 in Zong \etalabbr~\cite{zong16cascading} for the derivation of Lemma~\ref{lemma:probability_bound}.
Based on the above two lemmas, if we choose
\[
\alpha \geq \frac{1}{\sigma} \sqrt{d \log \left( 1 + \frac{nK}{d \sigma^2}\right) + 2 \log \left( n \right)} + \|\theta^\ast \|_2,
\]
we have $P(\bar{\cE}_{t})\leq 1/n$ for all $t$ and hence
\[
R^{\gamma}(n) \leq \frac{2\alpha K}{\gamma}   \sqrt{\frac{dn \log \left[ 1 + \frac{nK}{d \sigma^2}\right]}{\log \left( 1+ \frac{1}{\sigma^2}\right)}}+1.
\]
This concludes the proof for Theorem~\ref{theorem:main}. 